\documentclass[11pt]{article}
\pdfoutput=1
\usepackage{amsfonts, amssymb, url,
 algorithm,algorithmic, amsmath,amsthm,latexsym,color,epsfig,mathrsfs,enumerate,float}
\usepackage{graphics}
\usepackage{caption}
\usepackage{subcaption}
\usepackage{graphicx}
\theoremstyle{definition}
\newtheorem{definition}{Definition}[section]
\newtheorem*{definition*}{Definition}

\newtheorem{theorem}{Theorem}%[section]

\newtheorem{lemma}{Lemma}[section]
\newtheorem{claim}{Claim}
\newtheorem{conjecture}{Conjecture}

\numberwithin{subcase}{case}
\title{Lower bound for sensitivity of graph properties}
\author{Ilan Karpas \thanks{The Einstein Institute of Mathematics, 
The Hebrew University of Jerusalem, Jerusalem, Israel. Email: ilan.karpas@mail.huji.ac.il.}}
\date{}

\begin{document}

\maketitle

%Abstract
\begin{abstract}
We prove that the sensitivity of any non-trivial graph property on $n$ vertices is at least $\lfloor \frac{1}{2}n \rfloor$ , provided $n$ is sufficiently large.
\end{abstract}

\section{Introduction}
The notion of the  \textit{sensitivity complexity} of a boolean function (see definition in section $2$), or just \textit{sensitivity}, has been widely studied, since introduced by Cook and Dwork in \cite{Cook-Dwork}, where it was used to derive tight lower bounds for the running time of a PRAM computing boolean function $f$. Simon \cite{Simon} proved that if  all of $f$'s coordinates have positive influence, one has the lower bound of $s(f)=\Omega(\log{n})$. This bound is known to be asymptotically tight, as demonstrated by the so called \textit{monotone address function}. 

There have been attempts to find lower bounds for the sensitivity of certain classes of boolean functions. In particular, the class of \textit{weakly symmetric} functions has been studied. For this class, the best known lower bound is  $s(f)=\Omega(\log{n})$ for any non-constant weakly symmetric function $f$ in $n$ variables. Indeed, for weakly symmetric functions all coordinates have the same influence, from which it follows that all coordinates have positive influences if $f$ is also non-constant, and as we have already mentioned, this guarantees $s(f)=\Omega(\log{n})$.

Yet, so far, every non-constant weakly symmetric function $f$ found, has sensitivity $s(f)=\Omega(n^{1/3})$. Turan conjectured in \cite{Turan} that there exists some absolute constant $c>0$, so that $s(f)=\Omega(n^c)$ for any non-constant weakly symmetric boolean function.

Perhaps the main open problem regarding sensitivity, is its relation to the seemingly similar notion of \textit{block sensitivity} (see section $2$ for definition). Nisan and Szegedy posed the conjecture \cite{Nisan-Szegedy} in $1994$ that the two are polynomially equivalent.

This would have important implications, since block sensitivity is known to be polynomially equivalent to many important measures of the complexity of a function, including \textit{Decision tree complexity}, \textit{certificate complexity}, \textit{degree} as a real polynomial, \textit{Quantum query complexity} etc. Finding bounds on the separations between pairs of the aforementioned measures is a lively field  (see, for example, \cite{Ambainis, Beals, Nisan, Nisan-Szegedy}). While at the moment block sensitivity can only be shown to be exponentially upper bounded by sensitivity \cite{Kenyon-Kutin}, the largest known gap is quadratic, as shown first by \cite{Rubinstein}. For an excellent survey on the sensitivity conjecture, see \cite{Hatami}.

The sensitivity conjecture implies Turan's conjecture on weakly symmetric functions, because if $f$ is weakly symmetric, then $bs(f)=\Omega(n^\frac{1}{3})$  \cite{Sun2}.

Turan, in the same paper \cite{Turan}, studied the sensitivity of \textit{graph properties}. He showed that for any non-trivial graph property $f$ for $n$ vertices (and so, ${n \choose 2}$ boolean variables), $s(f)\geq \lfloor \frac{n}{4} \rfloor$. He conjectured:

\begin{conjecture} \cite{Turan} \label{Turan}
Let $f$ be a non-trivial graph property for graphs with $n$ vertices.
Then $s(f)\geq n-1$.
\end{conjecture}

 Turan's conjecture is the best possible, since the property "$G$ has a vertex of degree $n-1$" has sensitivity $n-1$. Wegener \cite{Wegener} proved $s(f)\geq n-1$ for any non-trivial \textit{monotone} graph property $f$, but made no improvement on Turan's lower bound for graph properties in general. This was done only in $2011$, when Sun \cite{Sun1} improved Turan's lower bound to $s(f)\geq \frac{6}{17}n$. In  \cite{Gao}, lower bounds for the sensitivity of bipartite graph-properties were obtained.

Our main result in this paper, is an improvemnt of the lower bound in \cite{Sun1}. We prove:

\begin{theorem}
\label{main theorem}
Let $f$ be a non-trivial graph property on $n$ vertices, for sufficiently large $n$. Then $s(f)\geq \lfloor \frac{1}{2}n \rfloor$.
\end{theorem}

The paper is organized as follows: in section $2$, we provide the relevant definitions and notations. The reader familiar with the topic may want to skip to subsection $2.2$, where we provide non-standard definitions used in this paper. In section $3$, we prove Theorem \ref{main theorem}, and in section $4$ we discuss related open problems.

\section{Preliminaries and Definitions}

\subsection{Standard Definitions}

\begin{definition}
Let $f:\{0,1\}^n \to \{0,1\}$ be a boolean function, and $x \in \{0,1\}^n$ some point. The \textit{sensitivity of $f$ at point $x$}, which we denote as $s(f,x)$, is
\begin{align*}
 s(f,x):=|\{i\in [n ]|f(x \oplus e_i)\neq f(x)\}|.
\end{align*}
In other words, it is the number of neighbours $y$ of $x$ in the hamming cube,
with $f(y)\neq f(x)$.

The \textit{sensitivity of f}, $s(f)$, is just the maximum sensitivity over all points:
\begin{align*}
s(f)=\max_x s(f,x).
\end{align*}
\end{definition}

\begin{definition}

Let $f:\{0,1\}^n \to \{0,1\}$ be a boolean function, and $x \in \{0,1\}^n$ some point. The \textit{ block sensitivity of $f$ at point $x$}, which we denote as $bs(f,x)$, is the maximum number $t$, such that there are $t$ pairwise disjoint subsets of $[n]$, $B_1,\dots ,B_t$, with the property $f(x)\neq f(x^{B_i})$ for all $i \in [t]$. Here, $x^{B_i}$ is obtained by flipping in $x$ every coordinate $j \in B_i$.

The \textit{ block sensitivity of f}, $bs(f)$, is just the maximum block sensitivity over all points:
\begin{align*}
bs(f)=\max_x bs(f,x).
\end{align*}
\end{definition}

Of course $s(f)\leq bs(f)$ for every boolean function $f$, because each block can be a sensitive coordinate.

\begin{definition}

A group $\Gamma \leq S_n$ is called \textit{transitive}, if for every $i,j \in [n]$, there exists an element $\pi \in \Gamma$, so that $\pi(i)=j$.

\end{definition}

\begin{definition}
Let $x=x_1\dots x_n $ be an $n$ bit boolean string, and let $\pi \in S_n$ be some permutation of $n$ elements. Then the \textit{action} of $\pi$ on $x$, which we denote by $\pi x$, is the string $\pi x:= x_{\pi(1)}\dots x_{\pi(n)}$.
\end{definition}

\begin{definition}
A boolean function $f: \{0,1\}^n \to \{0,1\}$ is called \textit{weakly symmetric}, if there is some transitive group $\Gamma \leq S_n$, so that for every $x \in \{0,1\}^n$ and every $\pi \in \Gamma$, $f(x)=f(\pi x)$. We say, in that case, that $f$ is \textit{closed under $\Gamma$}.
 \end{definition}

We now move on to discuss graph properties.

 For every $\{i,j\} \in {[n] \choose 2}$, define a suitable boolean variable $x_{\{i,j\}}\in \{0,1\}$. A string
$x \in \{0,1\}^{n \choose 2}$, then, can be considered simply as an assignment to all these variables. This set can also be identified with the set of all graphs with vertex set $[n]$, via the bijection

\begin{equation}
\label{bijection}
x \longrightarrow \Big([n], \big\{\{i,j\} \in {[n] \choose 2}\big|x_{i,j}=1\big\}\Big). 
\end{equation}

Consider the group $S_n$ acting on the set ${[n] \choose 2}$ in the following way: for $\{i,j\} \in {[n] \choose 2}$, and $\pi \in S_n$, we write  $\pi\{i,j\}:=\{\pi(i),\pi(j)\}$. Note that $S_n$ acts transitively on ${[n] \choose 2}$.

\begin{definition}
$f:\{0,1\}^{[n] \choose 2}\to \{0,1\}$ is called a \textit{graph property}, if
for every $\pi \in S_n$ and every point $x \in \{0,1\}^{[n] \choose 2}$, $f(x)=f(\pi x)$.
\end{definition}

 By abuse of notation, from now on we shall think of the domain of $f$ as the family of all graphs with vertex set $[n]$, using the bijection in \eqref{bijection}.
$f$ being a graph property, in this notation, means that $f(G)=f(H)$ whenever $G$ and $H$ are isomorphic graphs. Furthermore, we will sometimes think of a graph simply as its set of edges, so for example $|G|$ will denote the number of edges in $G$. This should not create confusion, since the vertex set is always understood to be $[n]$, unless explicitly stated otherwise.

\subsection{non-standard Definitions and notations}

Throughout this subsection, $f$ is a non-trivial graph property on graphs with vertex set $[n]$, and $f(\overline{K_n})=0$. 

\begin{definition} \label{positive}

For a graph $G=([n],E)$, $G \neq \overline{K_n}$, we define the \textit{positive minimum degree} of $G$, denoted by $\delta'(G)$, to be the minimum degree of any non-isolated vertex in $G$. We write $\delta'(\overline{K_n})=\infty$.\\
 For any natural number $k$, we denote by $G_{[k]}$ the maximal subgraph of $G$ (with vertex set $[n]$), for which $\delta'(G_{[k]})\geq k$.

Observe that if $\delta'(G) \geq k$, then $G_{[k]}=G$. Otherwise $G_{[k]} \subsetneq G$.
\end{definition}

\begin{definition}
Any graph $G$ can have three kinds of connected components: isolated vertices, components containing a cycle, and trees containing at least one edge.
We define the \textit{positive minimum tree component size}, which we denote by $c(G)$, to be the number of edges in the smallest connected component of $G$ that is a tree, \textbf{but not an isolated vertex}. If $G$ has no connected components of this form, we write $c(G)=\infty$.

For any natural number  $k$, we denote by $G_{(k)}$ the maximal subgraph of $G$ (with vertex set $[n]$), so that $c(G_{(k)})\geq k$.\\
Observe that if $c(G) \geq k$, then $G_{(k)}=G$. Otherwise $G_{(k)} \subsetneq G$.
\end{definition}

\begin{definition}
We call a graph $G$ \textit{minimal} with respect to $f$, if $f(G)=1$ and $f(G')=0$ for every graph $G' \subsetneq G$. The set of all minimal graphs for $f$ is denoted by $m(f)$. Notice that $m(f)\neq \emptyset$ because we took $f$ to be non-trivial.
\end{definition}

We write $\delta'(f)=\min_{G \in m(f)} \delta'(G)$, and $c(f)=\min_{G \in m(f)}c(G)$.

\begin{definition}
Let $T$ be a tree with $k$ edges. A \textit{tree construction sequence of $T$}, is a sequence of $k$ trees $T^{(1)},\dots ,T^{(k)}$ such that:
\begin{itemize}

\item  For every $i \in [k]$,  $T^{(i)}$ is a tree with $i$ edges.

\item $T^{(k)}$ is isomorphic to $T$.

\item For every $i \in [k-1]$, $T^{(i+1)}$ is obtained from $T^{(i)}$ by adding a new vertex, and connecting this vertex by an edge to some other vertex in $T^{(i)}$.
\end{itemize}
(See Figure~\ref{fig:Construction}).
\end{definition}

\vspace{0.5 cm}
\begin{figure} [ht] 
\centering
\vspace{0.5cm}
\includegraphics[width=0.5\textwidth]{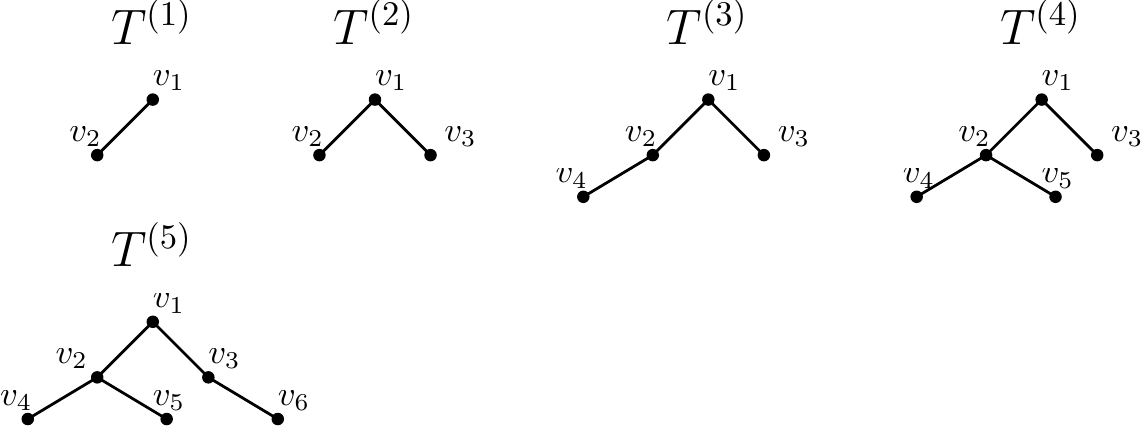}
\caption{Tree construction sequence} \label{fig:Construction}
\end{figure}

\section{proof of main theorem}

Before proving the main theorem, we need the following  lemma, found in \cite{Sun1}
\begin{lemma}
\label{deg1}
Let $f$ be a graph property, and $G$ be a graph. Let $v \in V(G)$ be a vertex of degree one, and $e=\{v,u\}$ the unique edge in $G$ where $v$ occurs. Then either $s(f) \geq |I(G)|+1$, or $f(G)=f(G \setminus e)$.
\end{lemma}

\begin{proof}
Let $f, G, v, e$ be as in the statement of the lemma. Assume that $f(G) \neq f(G \setminus e)$,  and let $I(G)=\{w_1,\dots ,w_{|I(G)|}\}$ . Observe  that   $G$ is isomorphic to $(G \setminus e) \cup \{u, w_i\}$ for every $1 \leq i \leq |I(G)|$. 
Thus,  $f(G ) = f((G \setminus e) \cup \{u, w_i\})$ for every $w_i \in I(G)$. But this means that $f(G \setminus e) \neq  f((G \setminus e) \cup \{u, w_i\})$, and as mentioned, $f(G \setminus e) \neq f(G)$. So the sensitivity of $f$ at $G \setminus e$ is at least $|I(G)|+1$, as claimed.
\end{proof}

\begin{proof}[proof of Theorem \ref{main theorem}]
Throughout the proof, we assume that all graphs have vertex set $[n]$.
Assume without loss of generality that for the graph with no edges $\overline{K_n}$, $f(\overline{K_n})=0$, and assume by contradiction that $s(f)< \lfloor \frac{n}{2} \rfloor$.

It is an easy, but important, observation, that if $G \in m(f)$ then $G$ is sensitive at any of the edges it contains, which means that

\begin{equation}
\label{mingraph size}
|G|\leq s(f,G)<\lfloor \frac{1}{2}n \rfloor
\end{equation}

As already mentioned, by our assumptions on $f$, $c(f)$ and $\delta'(f)$ are well defined. We divide our proof to three cases, based on their possible values:\\

\underline{\textbf{Case 1} ($\delta'(f)\geq 2$)}  
\newline \newline
 Let $G \in m(f)$, $\delta'(G)=\delta'(f)$.
 In this case, since there are no vertices of degree one in $G$, in particular any connected component which is not an isolated vertex, can not be a tree. Thus, for every connected component in $G$, the number of edges in that component is at least the number of vertices in it. Summing over all connected components, produces the inequality
 
\begin{equation}
\label{edges_isolated_vertices}
|G|+|I(G)|\geq n.
\end{equation}

Combining this with \eqref{mingraph size}

\begin{equation}
\label{isolated}
|I(G)|>\frac{1}{2}n
\end{equation}

We show an algorithm that finds a graph $G' \subsetneq G$, with $f(G')=1$. But we took $G \in m(f)$, which is a contradiction.
Denote the isolated vertices of $G$ by $I(G)=\{v_1,\dots,v_m\}$.

\begin{algorithm} [ht]
\caption{Algorithm for $\delta'(f)\geq 2$}
\begin{algorithmic}
\STATE $H \gets G$
\STATE $k \gets \delta'(G)$
\STATE $m \gets |I(G)|$
\STATE Find a vertex $v$ of degree $k$. Denote $N(v)=\{u,u_1,\dots,u_{k-1}\}$ 
\FOR{$i:=1$ to $k-1$} 
 \FOR {$j:=1$ to $m$}
   \IF {$\exists e \in G$ so that $f(H \setminus e)=f(H)$}
    \STATE $H \gets H\setminus e$
    \RETURN some minimal graph $G' \subset H$ \COMMENT{We prove that such $G'$ must exist}
   \ENDIF
   \STATE $H \gets H\cup \{u_i,v_j\}$
   \ENDFOR
 \ENDFOR
\STATE $H \gets H \setminus \{v,u\}$
\RETURN  some minimal graph $G' \subset H$ \COMMENT{ We prove that such $G'$ must exist}
\end{algorithmic}
\end{algorithm}

\pagebreak

\begin{claim}
\label{valuestays}
At any stage in the algorithm, $f(H)=1$. Thus $f(G')=1$.
\end{claim}
\begin{proof}
Assume by contradiction that this is not the case. Then there must be a first time during the execution of the algorithm when the value of $f(H)$ becomes zero. When we refer to $H$ in the proof, henceforth, we mean the graph $H$ in the algorithm immediately after its value becomes zero for the first time. This can occur either during some iteration of the for-loop, or at the end of the loop, while removing edge $\{v,u\}$.  \\\\
We deal with the latter case first:
Let $H^+=H\cup \{v,u\}$. That is, the graph in the algorithm just before removing edge $\{v,u\}$. By our assumption $f(H^+)=1$. But notice that $H \cup \{v_i,u\}$ is isomorphic to $H^+$, for all $1 \leq i \leq m$. Thus, $H$ is sensitive at all these edges, and at $\{v,u\}$, which means that \newline $s(f,H)\geq m+1 >\frac{1}{2}n$ , contradicting \eqref{isolated} (see Figure ~\ref{fig:out-loop}).\\

Next, we deal with the case that the value of $f(H)$ becomes zero for the first time during some iteration of the for-loop. This can not happen inside the if-condition, because the algorithm only removes edge $e$ inside the if-condition if guaranteed that removing it would not change the value of the function $f$. Hence, it could only happen when adding edge $\{u_i,v_j\}$, for some $1 \leq i \leq k-1$, $1 \leq j \leq m$. Denote  by $H^-=H \setminus \{u_i,v_j\}$ the graph just before adding this edge. Since the if-condition for $H^-$ did not hold, $H^-$ is sensitive at every edge contained in $G$. Furthermore, $H$ is isomorphic to $H^-\cup \{u_i,v_l\}$ for every $j \leq l \leq m$. That is, $H^-$ is also sensitive on any of these $m-j+1$ edges. Together, we see that
 
\begin{equation}
\label{H^-}
m-j+1+|G|\leq s(f,H^-)\leq s(f)<\frac{1}{2}n.
\end{equation}

Now let's turn our attention to $H$ itself. We know that $f(H) \neq f(H^-)$, and that $H^-$ is isomorphic to $H\setminus \{u_i, v_l\}$ for any $1 \leq l \leq j$. That is, $H$ is sensitive on all these edges. So 
\begin{equation}
\label{H}
j \leq s(f,H) <\frac{1}{2}n.
\end{equation}

Taking \eqref{H^-}+\eqref{H} gives $m+|G|+1<n$, which is a contradiction to \eqref{edges_isolated_vertices} (see Figure ~\ref{fig:in-loop}).
So $f(H)=1$, and since $G'$ is a minimal graph contained in $H$, by definition $f(G')=1$.
\end{proof}

\begin{claim}
\label{contained1}
$G' \subsetneq G$.
\end{claim}

\begin{proof}
Since $\delta'(G)=k$, we know that $G_{[k]}=G$ (see Definition \ref{positive}). On the otherhand, at each stage the algorithm adds to $H$ only edges to vertices in $I(G)$. However, for every such vertex $v_i \in I(H)$, the degree of $v_i$ \textbf{in} $\mathbf{H}$ remains strictly smaller than $k$ throughout the execution of the algorithm, which means that throughout the execution of the algorithm $H_{[k]} \subseteq G$. Finally, the algorithm removes an edge from $H$ that belongs to $G$, either inside the if-condition or after the for-loop. After this happens, no more changes to $H$ are made, and at this point $H_{[k]} \subsetneq G$. From the previous claim, $f(H)=1$. Since $G'$ is a minimal graph, by our assumption $\delta'(G')\geq k$. Thus, $G'=G'_{[k]} \subseteq H_{[k]} \subsetneq G$, proving the claim.
\end{proof}

\begin{figure} [ht]
    \centering

   \subcaptionbox {$f(H)$ changes outside loop \label{fig:out-loop}}{%
        \includegraphics[width=0.45 \textwidth]{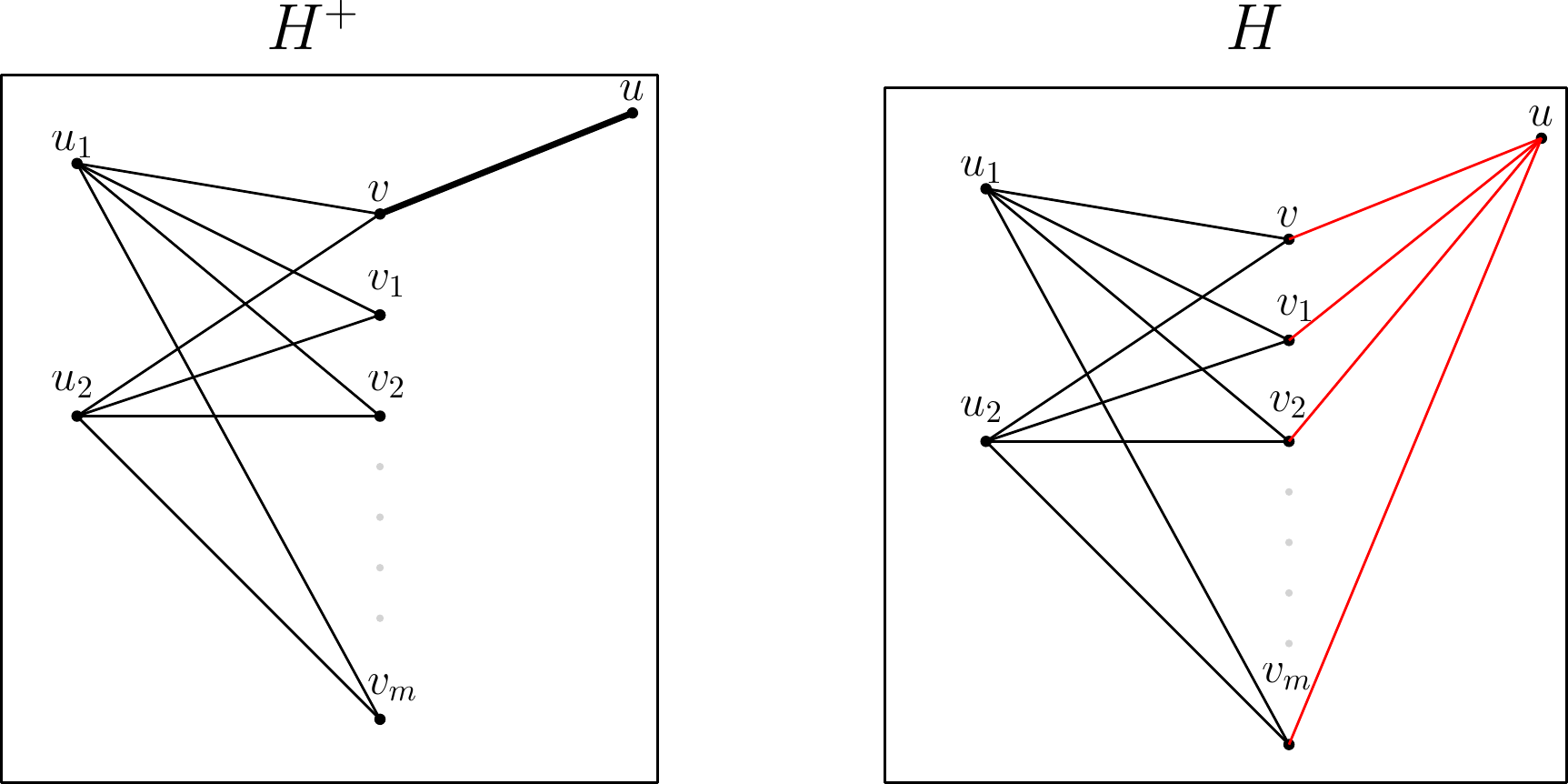}%
}  \par \medskip
      
    \subcaptionbox {$f(H)$ changes inside loop \label{fig:in-loop}}{%
        \includegraphics[width=0.45 \textwidth]{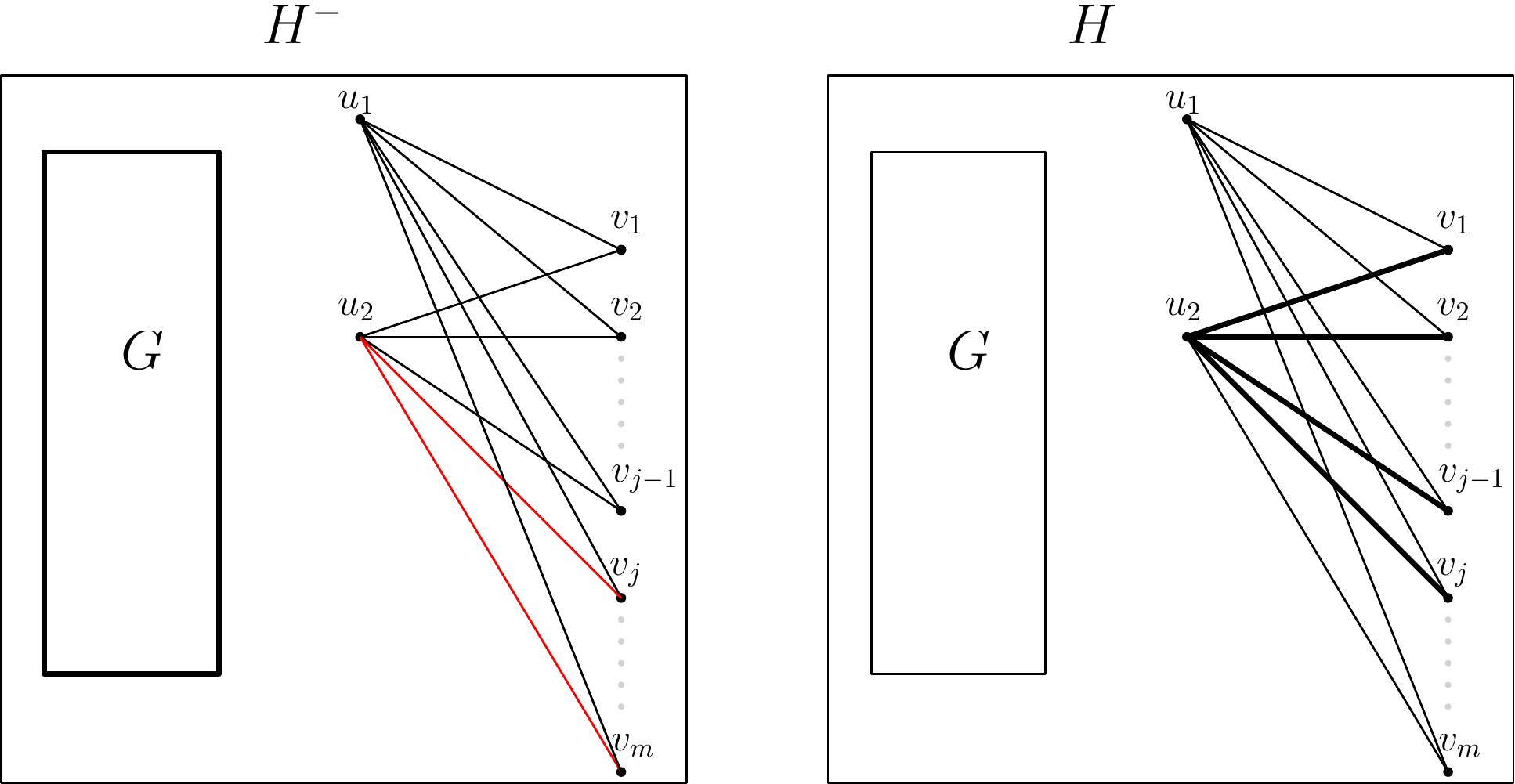}%
}     
\caption{Example with $\delta'(G)=3$. Black lines are edges, thick black lines are sensitive edges, red lines are sensitive non-edges.} 
\label{fig:Case1}
\end{figure}
\pagebreak

\underline{\textbf{Case 2} ($\delta'(G)=1$, $c(G)\geq 2$)}
 \newline \newline
Let $G \in m(f)$ be a graph so that $\delta'(G)=1$, $c(G)=c(f)$.
Denote by $T_1,\dots, T_r$ the connected components of $G$ that are trees but non isolated edges, by $C_1,\dots,C_l$ the connected components that are non-trees, and by $I(G)=\{v_1,\dots, v_m\}$ the set of isolated vertices. We do not assume a-priori that $c(G)<\infty$, that is, that $G$ has any components of the form $T_i$. However, we soon prove that this is indeed the case.

Notice, that for each component $T_i$ the number of edges in $T_i$ is exactly one less than the number of vertices in it, and for each component $C_i$ the number of edges in $C_i$ is at least the number of vertices in it. Summing over all components, we obtain

\begin{equation}
\label{edges2}
|I(G)|+|G|\geq n-r.
\end{equation}

We show that $r \geq 3$. Since $G$ is minimal, and $\delta'(G)=1$, by Lemma \ref{deg1} $|I(G)| \leq s(f)-1 \leq \frac{1}{2}n-2$. Yet, from that same minimality, $G$ is sensitive at every edge it contains, therefore \newline $|G|\leq s(f,G) \leq s(f) \leq \frac{1}{2}n-1$. Adding the two, we deduce $|I(G)|+|G| \leq n-3$, and so $r \geq 3$ by \eqref{edges2}.

From the pigeonhole principle and the minimality of $G$,

\begin{equation} \label{c(G)}
c(G)\leq |G|/r <\frac{1}{2r}n.
\end{equation}

In this particular case, we also assume $c(G)\geq 2$, so combined with \eqref{c(G)} this implies $r< \frac{1}{4}n$.
\\
 We provide an algorithm which has as input a minimal graph $G$ with the properties described above, and outputs a graph $G' \subseteq G$, $f(G')=1$, which is a contradcition to the minimality of $G$.
Recall, the set of isolated edges of $G$ is $I(G)=\{v_1,\dots, v_m\}$.

\begin{algorithm}

\caption{Algorithm for $\delta'(f)=1$, $c(f) \geq 2$}
\begin{algorithmic}
\STATE $H \gets G$
\STATE Find some tree component $T=(V_T,E_T)$ of $G$ with $c(G)=|E(T)|$, and some construction sequence $T^{(1)},\dots, T^{(c(G))}$ of $T$.
\FOR{$i:=1 \to c(G)-1$} 
   \IF {$\exists e \in G$ so that $f(H \setminus e)=f(H)$}
    \STATE $H \gets H \setminus e$
    \RETURN some minimal graph $G' \subset H$ \COMMENT{We prove that such $G'$ must exist}
  \ENDIF
 \STATE  connect vertex $v_{i+1}$ to some vertex $v_j$, $1 \leq j \leq i$, so that the subgraph of $H$ induced on vertices $v_1,\dots, v_{i+1}$ will be isomorphic to $T^ {(i)}$. 
\ENDFOR
\STATE Remove from $H$ some edge $e=\{u,v\}$ in component $T$, so that $(V_T\setminus \{u\},E_T \setminus \{e\})$ is isomorphic to $T^{(c(G)-1)}$.
\RETURN some minimal graph $G' \subset H$ \COMMENT{We prove that such $G'$ must exist}
\end{algorithmic}
\end{algorithm}

\begin{claim} 
\label{value_stays}
At each stage in the algorithm $f(H)$=1. Thus, $f(G')$=1.
\end{claim}
\begin{proof} 
Assume that at some stage in the algorithm $f(H)=0$, and let us take the first time that this happens. This stage can either occur inside the for-loop, or after the for-loop, while removing edge $e=\{u,v\}$.  Assume the former. The   value of $f(H)$ can not change inside the if-condition, because the if-condition only holds if edge $e$ could be removed from $H$ without changing the value of $f$ in $H$.

 From this it follows that the change occurs at some iteration, $i$, when adding the edge $\{v_{i+1},v_j\}$ for some $1 \leq j \leq i$. Denote by 
$H^-=H\setminus \{v_{i+1},v_j\}$ the graph just before this edge was removed.  $f(H^-)=1$. Furthermore, no edge $e \in G$ could be removed from $H^-$ without changing the value of the function, otherwise the if-condition in the beginning of the $i$th iteration would hold. Hence, $H^-$ is sensitive on all $|G|$ edges that belong to $G$. In addition to that, $H$ is isomorphic to $H^- \cup \{v_j,v_l\}$ for every $i+1 \leq l \leq m$. Thus, adding any one of these edges to $H^-$ changes the value of $f$ on the graph. Together, we see that \newline $s(f)\geq s(f,H^-)\geq m-i+|G|\geq m-c(G)+1+|G|$. From \eqref{edges2} and our assumption $s(f)\leq \frac{1}{2}n-1$, we obtain

\begin{equation}
\label{impossible1}
\frac{1}{2}n \geq n-r-c(G)+2,
\end{equation}

and recalling \eqref{c(G)}, we deduce

\begin{equation}
\frac{1}{2}n+2 <r+\frac{n}{2r}.
\end{equation}
This inequality does not hold for any $3 \leq r \leq n/4$, and $n$ sufficiently large.\\

Next, assume that $f(H)$ becomes zero for the first time after the loop, while removing edge $\{u,v\}$. 
Denote by $H^+=H \cup \{u,v\}$, the graph just before removing edge $\{u,v\}$ and changing the value of the function. The following facts are of importance to us:
\begin{itemize}

\item $H$ contains two connected components, which are isomorphic to $T^{(c(G))-1}$. Namely, one copy on vertices $v_1,\dots,v_{c(G)}$ with edges added by the algorithm, and one copy formed from $T$ by removing edge $\{u,v\}$. Call these copies $T_1$ and $T_2$, respectively.
\item $I(H)=\{v_{c(G)+1},\dots, v_m,u\}$, so 
\begin{equation}
\label{ind}
|I(H)|=m-c(G)+1.
\end{equation}
\end{itemize}

The graph $H^+$ is isomorphic to any graph of the form $H \cup \{v,w \}$,
$w \in I(H)$. Additionally, there is some vertex $v_j \in V(T_1)$, $1 \leq j \leq c(G)$, so that for any $w \in I(H)$, $H \cup \{v_j,w\}$ is isomorphic to $H^+$.
This is true, because $T_1$ and $T_2$ are isomorphic. From all of the above, it follows that $\frac{1}{2}n>s(f,H)\geq 2|I(H)|$. With \eqref{c(G)} and \eqref{ind}, this implies
\begin{equation}
\label{sens bound}
\frac{1}{2}n> 2(m-c(G)+1)\geq 2(m-|G|/r +1)
\end{equation}

From equation \eqref{edges2}, we know that $m\geq n-|G|-r$. Plugging this into \eqref{sens bound}, using  $|G| \leq \frac{1}{2}n-1$, we have

\begin{equation}
\label{Impossible}
\frac{1}{4}n>n-r-(\frac{1}{2}n-1)(1+\frac{1}{r})+1.
\end{equation}

Reogranizing, using $r \geq 3$:

\begin{equation}
r-2>\frac{1}{4}n(1-\frac{1}{2r}),
\end{equation}

or

\begin{equation}
r>n/4,
\end{equation}

which we have shown is impossible. Thus, at every stage of the algorithm $f(H)=1 $ (see Figure ~\ref{fig:outside-loop2}), and consequentially $f(G')=1$.
\end{proof}

\begin{claim}
\label{contained2}
$G' \subsetneq G$
\end{claim}

\begin{proof}
Let $c(G)=k$. Notice that throughout the algorithm, $H_{(k)} \subseteq G=G_{(k)}$, because all edges which are in $H\setminus G$ form a tree component in $H$ with less than $k$ edges. The last change the algorithm does to $H$, is removing an edge from $H$ that is also an edge of $G$. Thus, after this edge is removed and the algorithm does not change $H$ anymore, it is the case that $H_{(k)} \subsetneq G$, hence also $G'_{(k)} \subsetneq G$. Recall our choice of $G$: $G \in m(f)$ and $k=c(G)=c(f)$. $G' \in m(f)$ as well, so of course $c(G')\geq c(f)=k$, and therefore $G'=G'_{(k)} \subsetneq G$, proving the claim
\end{proof}

\begin{figure} [ht] 
\centering
\vspace{0.5cm}
\includegraphics[width=0.4\textwidth]{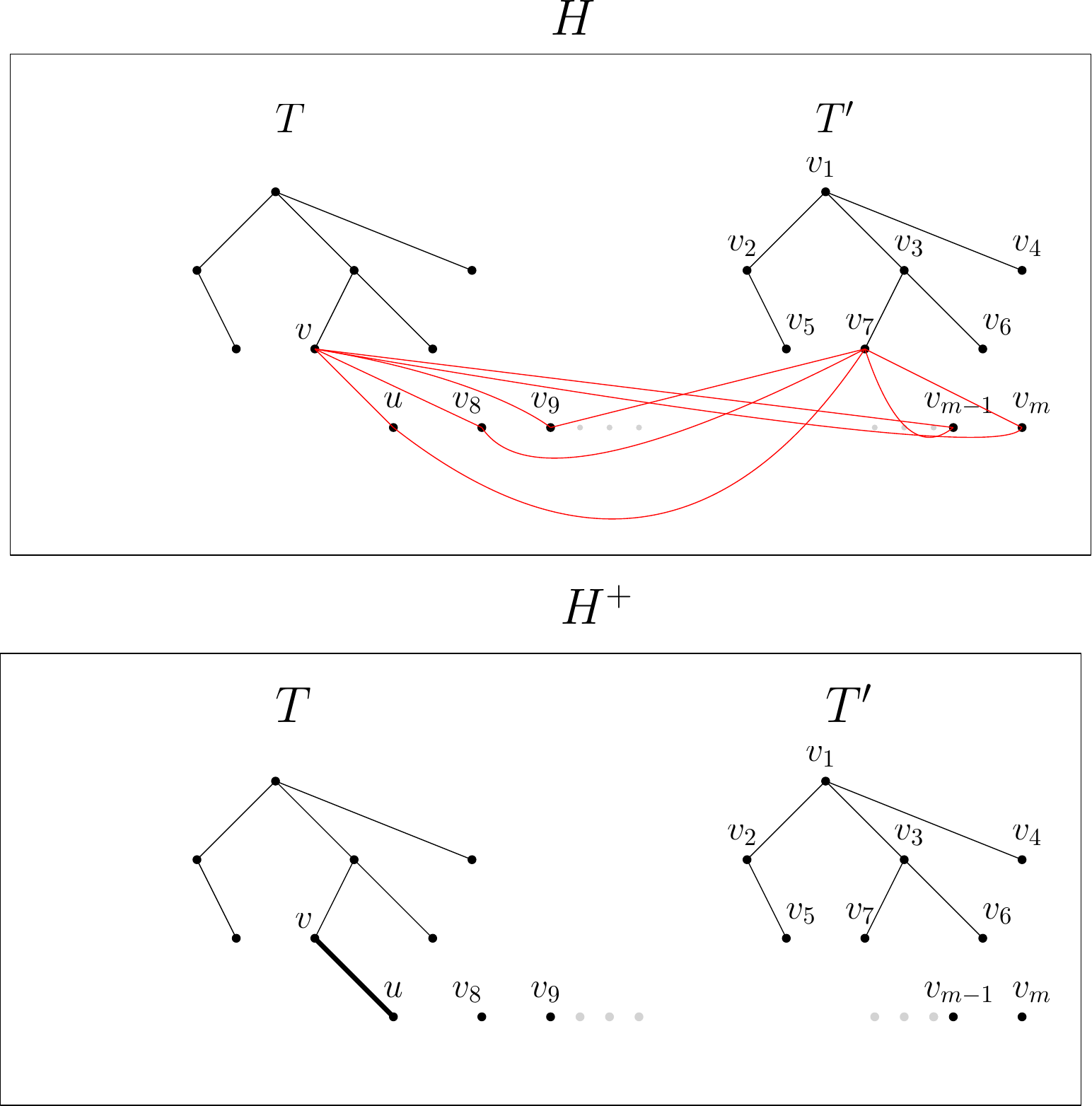}
\caption{Example with $c(G)=7$, $f(H)$ changes outside loop. Black lines are edges, thick black lines are sensitive edges, red lines are sensitive non-edges. } \label{fig:outside-loop2}
\end{figure}
\vspace{0.5cm}

\underline{\textbf{Case 3} ($c(f)=1$)}
\newline \newline
Choose a graph $G\in m(f)$, so that $\delta'(G)=1$, $c(G)=1$.  Let $E_1,\dots,E_m$ be its connected components with one edge, $C_1,\dots, C_l$ the connected components with at least $2$ edges, and $I(G)=\{w_1,\dots, w_r\}$  the set of isolated vertices. The following structural lemma will be useful:

\begin{lemma} 
\label{mingraphstructure}
For the graph $G$ above:
\begin{enumerate}
\item $r \geq 2$
\item $r=o(n)$
\item $m=\frac{1}{2}n-o(n)$ and $l=o(n)$

\end{enumerate}
\end{lemma}

\begin{proof} 
\begin{enumerate}

\item Since $G$ is minimal, $|G| \leq s(f,G) \leq \frac{1}{2}n-1$. In any graph, the number of non-isolated vertices is at most twice the number of edges. Thus, $r \geq n-2|G|\geq n-(n-2)=2$.

\item Let $\{u,v\} \in E(G)$ be an edge which connects two vertices of degree one. Remove this edge, to obtain graph $G^-=G \setminus \{u,v\}$. By minimality of $G$, $f(G^-)\neq f(G)$.  $|I(G^-)|=r+2$, and adding any edge between two isolated vertices in $G^-$ will create a graph isomorphic to $G$, and thus change the value of the function. Thus, ${r+2 \choose 2} \leq s(f,G^-) \leq s(f) \leq \frac{1}{2}n$, so $r=o(n)$.

\item Notice, that for each component $E_i$, $i \in [m]$, there are exactly two vertices in the component. Thus, the total number of vertices in these components is $2m$. For any component $C_i$, $i \in [l]$, $|V(C_i)| \leq |C_i|+1 \leq \frac{3}{2}|C_i|$, because $|C_i| \geq 2$. So we can write the following inequalites:

\begin{equation}
\label {edges}
\frac{1}{2}n-1 \geq |G|=m+\sum_{i=1}^l |C_i|
\end{equation}

\begin{equation}
\label{vertices} 
n=r+2m+\sum_{i=1}^l (|C_i|+1) \geq o(n)+2m+ \frac{3}{2}\sum_{i=1}^{l}|C_i|.
\end{equation}

taking $4\cdot \eqref{edges} -2\cdot \eqref{vertices}$, we obtain

$\sum_{i=1}^l |C_i|=o(n)$, and thus $m=\frac{1}{2}n-o(n)$, proving $3$.
\end{enumerate}
\end{proof}

Our method of proof will be to obtain two isomorphic graphs $H$ and $H'$, for which $f(H) \neq f(H')$, which is impossible since $f$ is a graph property. Let $E_i=\{v_i,u_i\}$, for $1 \leq i \leq m$. We define $H=G\cup \{u_1,u_2\} \cup  (\bigcup_{i=3}^{\lceil n/6 \rceil+1}\{v_1,v_i\}$. Next, we define $H'=H \cup \{v_1,v_2\} \setminus\{u_1,v_1\}$. Notice, that $H'$ is isomorphic to $H$, since the function $\pi:V(H') \to V(H)$ given by $\pi(v_2)=u_1$, $\pi(u_1)=v_2$, $\pi(x)=x$ for all $x\in V\setminus \{v_2,u_1\}$ is a graph isomorphism.
So to finish the proof, we just need to prove $f(H) \neq f(H')$.

\begin{claim}
\label{H=1}
$f(H)=1$.
\end{claim}
\begin{proof}
Construct $H$ from $G$ in the following way. First, add edge $\{u_1, u_2\}$, and then add the edges $\{v_1,v_i\}$ one by one, starting from $i=3$ and increasing $i$ by one each time. In this way we obtain a sequence of graphs $G:=G_1,G_2,G_3,\dots, G_{\lceil n/6 \rceil+1}=H$, so that $G_2=G \cup \{u_1,u_2\}$, and for \newline $3 \leq i \leq \lceil n/6 \rceil +1$, $G_i=G_{i-1} \cup \{v_1,v_i\}$.

 Notice that $G_2$ is isomorphic to $G \cup \{u_i,u_j\}$ for any $1 \leq i<j \leq m$, and so if $G$ was sensitive for $\{u_1,u_2\}$ it would be sensitive for $\Omega(m^2)=\Omega(n^2)$ edges, which is impossible for sufficiently large $n$. So $f(G_1)=1$. Assume that $G_i$ is the first graph in the sequence for which $f(G_i)=0$,  for some $3 \leq i \leq \lceil n/6 \rceil+1$. That is, $G_{i-1}$ is sensitive on edge $\{v_1,v_i\}$.

 However, all graphs of the form $G_{i-1}\cup \{v_1,v_j\}$ and $G_{i-1} \cup \{v_1,u_j\}$ are isomorphic to $G_i$, for $i \leq j \leq m$. Hence, $s(f,G_{i-1})\geq 2(m-i+1) \geq 2(m-\lceil n/6\rceil)=\frac{2}{3}n-o(n)>\frac{1}{2}n$, for sufficiently large $n$. This contradicts our assumption $s(f)\leq \frac{1}{2}n-1$. Hence, $f(H)=f(G)=1$.
\end{proof}

\begin{claim}
\label{H'=0}
$f(H')=0$.
\end{claim}
\begin{proof}
Since $G$ is minimal, $f(G \setminus \{u_1,v_1\})=0$. Denote by $G'=G \setminus \{u_1,v_1\}$. $u_1,v_1 \in I(G')$, because they both belong to $E_1$ in $G$. Thus, $|I(G')|=r+2 \geq 4$, from $1$ in lemma \ref{mingraphstructure}. Construct $H'$ from $G'$, by first adding edges $\{v_1,v_i\}$ one by one, starting from $i=2$ and increasing $i$ by one each time. Finally, add the edge $\{u_1,u_2\}$.

 In this way we obtain a sequence of graphs $G':=G'_1,G'_2,\dots,G'_{\lceil n/6 \rceil +1}, H'$, so that $G'_i=G'_{i-1} \cup \{v_1,v_i\}$ for every $2 \leq i \leq \lceil n/6 \rceil +1$, and $H=G'_{\lceil n/6 \rceil +1}\cup \{u_1,u_2\}$. From identical consideration as in the previous claim, $f(G_i)=f(G_{i-1})=0$ for all $2 \leq i \leq \lceil n/6 \rceil +1$, since otherwise $G_{i-1}$ would be sensitive for too many edges.

 Thus, we just need to prove $f(H')=f(G_{\lceil n/6 \rceil +1})$. Denote by $H'^-=G'_{\lceil n/6 \rceil +1}$. the only vertex that is isolated in $G'$ but not in $H'^-$ is $v_1$. So let $u_1,w_1,w_2 \in I(H'^-)$. Notice that $H'$ is isomorphic to all graphs of the form $H'^- \cup \{u_1,u_j \}$, $H'^- \cup \{w_1,u_j\}$, $H'^- \cup \{w_2,u_j\}$, for $2 \leq j \leq \lceil n/6 \rceil +1$. This makes a total of $3\lceil n/6 \rceil>\frac{1}{2}n-1\geq s(f)$, so if $f(H'^-)\neq f(H')$, then $s(f,H'^-)>s(f)$, which is impossible. Hence, $f(H')=0$, proving the claim.
\end{proof}

\end{proof}

\section{Concluding Remarks}
Conjecture \ref{Turan} remains unresolved, more than thirty years after it was first formulated. This author remains agnostic about the veracity of the conjecture.

One can of course ask an analogous question for $k$-uniform hypergraphs for any 
fixed integer $k\geq 2$. 
That is, is the minimum sensitivity for any non-trivial $k$-uniform hypergraph property the same as the minimum sensitivity for any non-trivial \textbf{monotone} $k$-uniform hypergraph property. It turns out, though, that for $k>2$ this is false, even asymptotically.

Indeed, it is well known, and not hard to see, that if $f$ is a non-trivial monotone $k$-uniform hypergraph property, then $s(f)=\Omega(n^{k/2})$. 
On the other hand, in a paper recently uploaded to arxiv, Li and Sun \cite{Li}
show that for any $k \geq 2$, there exists a non-trivial $k$-uniform hypergraph property so that
\begin{equation}
s(f)=O(n^{\lceil k/3 \rceil}),
\end{equation}

 and if $k \equiv 1 (\textrm{mod}\ 3)$, this bound can be improved to
\begin{equation}
s(f)=O(n^{\lceil k/3 \rceil-1/2}).
\end{equation}

Regardless, bounds for the sensitivity of $k$-uniform hypergraph properties remain of interest, both for the monotone and for the general case.

Finally, we propose a weaker form of Conjecture \ref{Turan}, where we limit $f$ to be a non-trivial \textit{min-term} graph property. See \cite{Chakraborty} for the definition of a min-term function.

\begin{conjecture} \label{min term graph properties}
Let $f:\{0,1\}^{n \choose 2} \to \{0,1\}$ be a non-trivial \textit{min-term} graph property. Then  
\begin{equation*}
s(f)\geq n-1.
\end{equation*}

\end{conjecture}

\end{document}